\documentclass[letterpaper,11pt]{article}
\usepackage[utf8]{inputenc}

\pdfoutput=1

\usepackage{amsmath, amsthm, amssymb}
\usepackage[noend]{algpseudocode,algorithm}
\usepackage{algorithmicx}
\usepackage{mathtools}
\usepackage[numbers]{natbib}
\usepackage{comment} 
\usepackage{color-edits} 
\usepackage[most]{tcolorbox}
\usepackage{xfrac}
\usepackage{hyperref}
\usepackage{multirow}
\usepackage{caption}
\usepackage{bm}
\usepackage{newfloat}
\usepackage{enumitem}
\usepackage{xpatch}

\makeatletter
\xpatchcmd{\algorithmic}
  {\ALG@tlm\z@}{\leftmargin\z@\ALG@tlm\z@}
  {}{}
\makeatother

\usepackage[margin=1in]{geometry}

\allowdisplaybreaks

\definecolor{mygreen}{RGB}{20,120,60}

\title{Exponentially Faster Massively Parallel Maximal Matching\footnote{A preliminary version appeared in the proceedings of the \emph{60th annual IEEE Symposium on Foundations of Computer Science }(FOCS 2019).\vspace{0.2cm}}}

\author{Soheil Behnezhad\thanks{Khoury College of Computer Sciences, Northeastern University.}\\ \texttt{s.behnezhad@northeastern.edu}
\and MohammadTaghi Hajiaghayi\footnotemark[3]\\ \texttt{hajiagha@cs.umd.edu}
\and David G. Harris\thanks{Department of Computer Science, University of Maryland.}
\\ \texttt{davidgharris29@gmail.com}
}

\date{}

\newcommand{\MPC}[0]{\ensuremath{\mathsf{MPC}}}
\newcommand{\local}[0]{\ensuremath{\mathsf{LOCAL}}}

\newcommand{\PRAM}[0]{\ensuremath{\mathsf{PRAM}}}

\newcommand{\E}[0]{\ensuremath{\mathbb{E}}}

\newcommand{\greedymatching}[1]{\ensuremath{\mathsf{GreedyMM}(#1)}}

\newcommand{\edgeoracle}[0]{\ensuremath{\mathcal{EO}}}

\newcommand{\degreeoracle}[0]{\ensuremath{\mathcal{DO}}}

\newcommand{\yes}[0]{\ensuremath{\textsc{yes}}}
\newcommand{\no}[0]{\ensuremath{\textsc{no}}}

\DeclareMathOperator{\poly}{poly}
\DeclareMathOperator{\polylog}{polylog}
\DeclareMathOperator{\polyloglog}{polyloglog}
\DeclareMathOperator{\res}{rdeg}

\renewcommand{\epsilon}{\varepsilon}

\DeclareMathOperator{\var}{Var}

\addauthor{sb}{blue}    
\addauthor{dgh}{red}    

\newtheorem{theorem}{Theorem}

\newtheorem{lemma}{Lemma}[section]

\newtheorem{proposition}[lemma]{Proposition}
\newtheorem{corollary}[lemma]{Corollary}

\newtheorem{definition}[lemma]{Definition}
\newtheorem{claim}[lemma]{Claim}

\newcommand{\smparagraph}[1]{
\vspace{0.3cm}
\noindent \textbf{#1}
}

\definecolor{mygreen}{RGB}{20,125,20}
\definecolor{linkcolor}{RGB}{0,0,230}
\definecolor{mylightgray}{RGB}{230,230,230}

\hypersetup{
     colorlinks=true,
     citecolor= mygreen,
     linkcolor= mygreen,
     urlcolor= mygreen
}

\newcommand{\etal}[0]{\textit{et al.}}

\algnewcommand{\IIf}[1]{\State\algorithmicif\ #1\ \algorithmicthen}
\algnewcommand{\EndIIf}{\unskip\ \algorithmicend\ \algorithmicif}

\newcounter{myalgctr}

\newenvironment{tbox}{
\vspace{0.2cm}
\begin{tcolorbox}[width=\textwidth,
                  enhanced,
                  boxsep=2pt,
                  left=1pt,
                  right=1pt,
                  top=4pt,
                  boxrule=1pt,
                  arc=0pt,
                  colback=white,
                  colframe=black,
                  unbreakable
                  ]
}{
\end{tcolorbox}
}

\newenvironment{graytbox}{
\vspace{0.2cm}
\begin{tcolorbox}[width=\textwidth,
                  enhanced,
                  frame hidden,
                  boxsep=6pt,
                  left=1pt,
                  right=1pt,
                  top=4pt,
                  boxrule=1pt,
                  arc=0pt,
                  colback=mylightgray,
                  colframe=black,
                  breakable
                  ]
}{
\end{tcolorbox}
}

\newcommand{\tboxhrule}[0]{\vspace{0.1cm} \hrule \vspace{0.2cm}}

\newenvironment{titledtbox}[1]{\begin{tbox}#1 \tboxhrule}{\end{tbox}}

\newenvironment{tboxalg}[1]{\refstepcounter{myalgctr}\begin{titledtbox}{\textbf{Algorithm \themyalgctr.} #1}}{\end{titledtbox}}

\begin{document}
\maketitle

\begin{abstract}
\setlength{\parskip}{0.4em}
The study of approximate matching in the {\em Massively Parallel Computations} (\MPC{}) model has recently seen a burst of breakthroughs. Despite this progress we still have a limited understanding of {\em maximal matching} which is one of the central problems of parallel and distributed computing. All known \MPC{} algorithms for maximal matching either take polylogarithmic time which is considered inefficient, or require a strictly super-linear space of $n^{1+\Omega(1)}$ per machine.

In this work, we close this gap by providing a novel analysis of an extremely simple algorithm, which is a variant of an algorithm  conjectured to work by \citet*{DBLP:conf/stoc/CzumajLMMOS18}. The algorithm edge-samples the graph, randomly partitions the vertices, and finds a random greedy maximal matching within each partition. We show that this algorithm drastically reduces the vertex degrees. This, among other results, leads to an $O(\log \log \Delta)$ round algorithm for maximal matching with $O(n)$ space (or even {\em mildly sublinear} in $n$ using standard techniques).

As an immediate corollary, we get a $2$ approximate {\em minimum vertex cover} in essentially the same rounds and space, which is the optimal approximation factor under standard assumptions. We also get an improved $O(\log\log \Delta)$ round algorithm for $1 + \varepsilon$ approximate matching. All these results can also be implemented in the {\em congested clique} model in the same number of rounds.
\end{abstract}
\clearpage

\section{Introduction}

The success of modern parallel frameworks such as MapReduce~\cite{DBLP:journals/cacm/DeanG08}, Hadoop~\cite{DBLP:books/daglib/0025439}, or Spark~\cite{DBLP:conf/hotcloud/ZahariaCFSS10}  has led to an active line of research for understanding the true computational power of such systems. The {\em Massively Parallel Computations} (\MPC{}) model provides a clean abstraction of these frameworks and has become the standard theoretical model for this purpose (see Section~\ref{sec:mpc} for the model).

In this work, we consider the {\em maximal matching} problem in the \MPC{} model. It is one of the most fundamental graph problems in parallel and distributed computing with far-reaching practical and theoretical implications. The study of maximal matching can be traced back to \PRAM{} algorithms of 1980s \cite{DBLP:conf/stoc/Luby85, DBLP:journals/ipl/IsraelI86, DBLP:journals/jal/AlonBI86} and has been studied in various computational models since then.

In the \MPC{} model, maximal matching is particularly important; an algorithm for it directly gives rise to algorithms for $1+\varepsilon$ approximate {\em maximum matching}, $2+\varepsilon$ approximate {\em maximum weighted matching}, and $2$ approximate {\em minimum vertex cover} with essentially the same number of rounds and space. Each of these problems has been studied on its own \cite{DBLP:conf/stoc/CzumajLMMOS18, DBLP:journals/corr/abs-1709-04599, DBLP:conf/podc/GhaffariGKMR18, assadiedcs, DBLP:conf/spaa/AssadiK17, DBLP:conf/spaa/BehnezhadDETY17, DBLP:journals/corr/abs-1807-08745, DBLP:conf/spaa/AhnG15}.

\paragraph{Known bounds.} For many graph problems, including maximal matching, there are $O(\log n)$ round \MPC{} algorithms coming from straightforward simulation of \PRAM{} algorithms \cite{DBLP:conf/stoc/Luby85, DBLP:journals/ipl/IsraelI86, DBLP:journals/jal/AlonBI86}. This can be improved to $\widetilde{O}(\sqrt{\log \Delta})$ rounds via simulation of distributed LOCAL algorithms \cite{ghaffariuitto}.  The main goal, however, is to obtain significantly faster (i.e., subpolylogarithmic round) algorithms by further utilizing \MPC{}'s additional powers.  

Currently, the only known such algorithm for maximal matching is that of Lattanzi \etal~\cite{DBLP:conf/spaa/LattanziMSV11} which requires $O(1/\delta)$ rounds using $O(n^{1+\delta})$ space. Their algorithm's round complexity, however, blows up back to $\Theta(\log n)$ as soon as memory becomes $O(n)$. In comparison, due to a breakthrough of \citet*{DBLP:conf/stoc/CzumajLMMOS18}, we have algorithms for $1+\varepsilon$ approximate matching that take $O(\log \log n)$ rounds using $O(n)$ space \cite{DBLP:conf/stoc/CzumajLMMOS18, DBLP:conf/podc/GhaffariGKMR18, assadiedcs}. Unfortunately, this progress on approximate matching offers no help for maximal matching or related problems. In fact,  these algorithms also require up to $\Omega(\log n)$ rounds to maintain maximality.

We also mention an MPC algorithm of \citet*{ghaffari2018improved} to obtain a maximal independent set (MIS) of a graph in $O(\log \log n)$ rounds using $\tilde O(n)$ space per machine. This algorithm could be applied to the line graph of $G$, which has $m$ vertices, but it would then require $\tilde O(m)$ space per machine --- significantly larger than the desired bound $O(n)$.

\paragraph{Our contribution.} In this paper, we give \MPC{} algorithms for maximal matching that are exponentially faster than the state-of-the-art (we describe our precise results in Section~\ref{sec:results}). We achieve this by providing a novel analysis of an extremely simple and natural algorithm: namely, the algorithm edge-samples the graph, randomly partitions the vertices into disjoint subsets, and finds a greedy maximal matching within the induced subgraph of each partition. (See Algorithm~\ref{alg:nearmaximal} in Section~\ref{sec:roadmap} for formal details.) It then commits the edges of this greedy matchings to the final output, thereby simplifying the residual graph. 

This partitioning is useful since each induced subgraph can be sent to a different machine. Czumaj \etal{}~\cite{DBLP:conf/stoc/CzumajLMMOS18} had conjectured that a variant of this algorithm might work and left its analysis as one of their main open problems:\footnote{A more detailed variant of the algorithm was also described in the following TCS+ talk by Artur Czumaj (starts from 1:03:23): \url{https://youtu.be/eq0jwAnJu9c?t=3803}.}
\begin{quotation}
\noindent ``{\em Finally, we suspect that there is a simpler algorithm for the problem {\normalfont [...]} by simply greedily matching high-degree vertices on induced subgraphs {\normalfont [...]} in every phase. Unfortunately, we do not know how to analyze this kind of approach.}'' \cite{DBLP:conf/stoc/CzumajLMMOS18}
\end{quotation}

We are not able to show that the simple process proposed by Czumaj \etal{}~\cite{DBLP:conf/stoc/CzumajLMMOS18} would work directly. However, we show something almost as powerful: the process significantly reduces the total number of edges of the residual graph \emph{in expectation}.  By adding a few post-processing steps (see Algorithm~\ref{alg:leftoververtices} in Section~\ref{sec:leftoververtices} later), we can ensure that the maximum degree of the residual graph gets reduced significantly and with high probability.

We summarize our results and their implications in Section~\ref{sec:results} and  give a high-level overview of the analysis in Section~\ref{sec:highlevel}.

\subsection{Main Results}\label{sec:results}

\begin{graytbox}
\begin{theorem}[main result]\label{thm:main}
	Given a graph $G$ with $n$ vertices, and $m$ edges and max degree $\Delta$, there is a randomized \MPC{} algorithm to compute a maximal matching that	\begin{enumerate}[label={(\arabic*)}, itemsep=0ex]
		\item takes $O(\log \log \Delta)$ rounds using $O(n)$ space per machine, 
		\item or takes $O(\log \frac{1}{\delta})$ rounds using $O(n^{1+\delta})$ space per machine, for any parameter $\delta \in (0, 1)$.
		\end{enumerate}
	The algorithm succeeds w.e.h.p.\footnote{We say an event occurs \emph{with exponentially high probability} (w.e.h.p.) if it occurs with probability $1 - e^{-n^{\Omega(1)}}$.} and requires an optimal total space of $O(m)$.
\end{theorem}
\end{graytbox}

Theorem~\ref{thm:main} part (1) provides the first subpolylogarithmic round \MPC{} algorithm for maximal matching that does not require super-linear space in $n$. In fact, it improves exponentially over the prior algorithms in this regime  \cite{DBLP:conf/stoc/Luby85,DBLP:conf/spaa/LattanziMSV11,ghaffariuitto}. Furthermore, Theorem~\ref{thm:main} part (2) exponentially improves over Lattanzi \etal{}'s algorithm \cite{DBLP:conf/spaa/LattanziMSV11} which requires $O(1/\delta)$ rounds using $O(n^{1+\delta})$ space.

\begin{theorem}\label{thm:sublinearspace}
Given a graph $G$ with $n$ vertices and $m$ edges and max degree $\Delta$, there is an \MPC{} algorithm to compute a maximal matching in $O(\log \log \Delta + \log \log \log n)$ rounds and $n/2^{\Omega(\sqrt{\log n})}$ space per machine. The algorithm succeeds w.e.h.p. and uses total space $O(m + n^{1+\gamma})$ for any constant $\gamma > 0$.
\end{theorem}

Theorem~\ref{thm:sublinearspace} further improves the space per machine to {\em mildly sublinear} with the same round complexity (ignoring the lower terms). Note that the $n/2^{\Omega(\sqrt{\log n})}$ space usage here goes below the $n/\polylog n$ space that has commonly been considered for problems such as approximate matching \cite{DBLP:conf/stoc/CzumajLMMOS18, DBLP:conf/podc/GhaffariGKMR18, assadiedcs} and graph connectivity \cite{DBLP:journals/corr/abs-1805-02974}.


\paragraph{Other implications.} It is well-known that the set of matched vertices in a maximal matching is a 2-approximation of minimum vertex cover. It is also closely connected to problems of maximum matching.  As such, our results can be applied to these and a few other related problems.

\begin{corollary}\label{cor:vertexcover}
All algorithms of Theorems~\ref{thm:main} and \ref{thm:sublinearspace} can be applied to the 2-approximate minimum vertex cover problem as well.
\end{corollary}

The problem of whether an approximate vertex cover can be found faster in \MPC{} with $O(n)$ space was first asked by Czumaj \etal{}~\cite{DBLP:conf/stoc/CzumajLMMOS18}. Subsequent works showed that indeed $O(\log \log n)$ algorithms are achievable and the approximation factor has been improved from $O(\log n)$ to $O(1)$ to $2+\varepsilon$ \cite{DBLP:journals/corr/abs-1709-04599, DBLP:conf/podc/GhaffariGKMR18, assadiedcs}. Note that, under the Unique Games Conjecture, Corollary~\ref{cor:vertexcover} gives the optimal approximation ratio for polynomial-time algorithms \cite{khot2008vertex}.\footnote{It is a standard assumption in \MPC{} algorithms that each machine should run a polynomial-time algorithm (see \cite{DBLP:conf/soda/KarloffSV10, DBLP:conf/stoc/AndoniNOY14}). If we remove this assumption, it may be possible to improve the approximation ratio for vertex cover.}

\begin{corollary}\label{cor:approximatematching}
For any arbitrary constant $\varepsilon > 0$, Theorem~\ref{thm:main} can be used to give algorithms for $1+\varepsilon$ approximate matching and $2+\varepsilon$ approximate maximum weighted matching in asymptotically the same number of rounds and space.
\end{corollary}

The reduction from maximal matching (and in fact, any $O(1)$ approximate matching) to $1+\varepsilon$ approximate matching is due to McGregor \cite{DBLP:conf/approx/McGregor05} (see \cite{assadiedcs}) and the reduction to $2+\varepsilon$ approximate weighted matching is due to  Lotker et al.~\cite{DBLP:journals/siamcomp/LotkerPR09} (see \cite{DBLP:conf/stoc/CzumajLMMOS18}). We also note that if the space is $\tilde O(n)$, then our algorithm can be used in a framework of Gamlath \etal{}~\cite{gamlath2018weighted} to get an $O(\log \log \Delta)$ round algorithm for $1+\epsilon$ approximate maximum weighted matching. Corollary~\ref{cor:approximatematching} also strengthens the round-complexity of the results in \cite{DBLP:conf/stoc/CzumajLMMOS18, DBLP:conf/podc/GhaffariGKMR18, assadiedcs} from $O(\log \log n)$ to $O(\log \log \Delta)$ using $O(n)$ space.\footnote{The algorithms of \cite{DBLP:conf/stoc/CzumajLMMOS18, DBLP:conf/podc/GhaffariGKMR18, assadiedcs} appear to still require $\Omega(\log \log n)$ rounds even for $\Delta = \poly\log n$ since they switch to an $O(\log \Delta)$ round algorithm at this threshold. Corollary~\ref{cor:approximatematching}, however, gives an $O(\log \log \log n)$ round algorithm on such graphs.}

Finally, there is a close connection between the MPC model and congested clique model \cite{DBLP:journals/corr/abs-1802-10297, DBLP:conf/podc/Lenzen13}, leading to the following results:
\begin{corollary}\label{cor:congestedclique}
Theorem~\ref{thm:main} directly gives an $O(\log \log \Delta)$ round algorithm for maximal matching in the congested clique model. It also leads to $O(\log \log \Delta)$ round congested clique algorithms for 2-approximate vertex cover, $1+ \varepsilon$ approximate maximum matching, and $2 + \varepsilon$ approximate maximum weighted matching by known reductions.
\end{corollary}

The problem of maximal matching in the congested clique model was first posed by Ghaffari~\cite{DBLP:conf/podc/Ghaffari17}, with a $\widetilde{O}(\sqrt{\log \Delta})$-round algorithm appearing in \citet*{ghaffariuitto}. Corollary~\ref{cor:congestedclique} exponentially improves over this bound.

\smparagraph{Recent developments.} After the conference version of the present paper, \citet*{AssadiLT21} used our result to give an $O(1/\epsilon^2 \cdot \log \log n)$-round MPC algorithm for $1 + \epsilon$ approximate matching in bipartite graphs, and similarly \citet*{fischer2022deterministic} used it to give an $O(\poly(1/\epsilon) \cdot \log \log n )$-round MPC algorithm for $1 + \epsilon$ approximate matching in general graphs. These improve the $\epsilon$-dependency of the previous $(1/\epsilon)^{O(1/\epsilon)} \log \log n$ round algorithm for general graphs based on McGregor's reduction \cite{DBLP:conf/approx/McGregor05}. Note that, for this purpose, it is essential to obtain a maximal matching as in Theorem~\ref{thm:main}.

\subsection{High Level Technical Overview}\label{sec:highlevel}
As discussed above, if the space per machine is $n^{1+\Omega(1)}$, then we already know how to find a maximal matching efficiently \cite{DBLP:conf/spaa/LattanziMSV11}. Such previous algorithms use ideas such as edge-sampling the graph into a \emph{single} machine \cite{DBLP:conf/spaa/LattanziMSV11,DBLP:conf/spaa/BehnezhadDETY17,DBLP:conf/spaa/AhnG15}, and these require $\Omega(\log n)$ rounds if the space is $O(n)$. Vertex partitioning \cite{DBLP:conf/soda/KarloffSV10, DBLP:conf/nips/BateniBDHKLM17, DBLP:conf/stoc/CzumajLMMOS18, assadiedcs, DBLP:conf/podc/GhaffariGKMR18}, which in the context of matching was first used by \cite{DBLP:conf/stoc/CzumajLMMOS18}, helps to make use of multiple machines. The general idea is to randomly partition the vertices and each machine can separately find a matching in the induced subgraph of a partition.

Algorithms in this framework may make different choices for the internal matching algorithm on these induced subgraphs \cite{DBLP:conf/stoc/CzumajLMMOS18, assadiedcs, DBLP:conf/podc/GhaffariGKMR18}. \emph{Greedy maximal matching} is one of the simplest matching algorithms: it iterates over each edge according to some given ordering $\pi$, and adds it to the matching if none of its incident edges are part of the matching so far. In other words, it is the lexicographically-first MIS of the line graph of $G$. This matching turns out to have several desirable structural properties that make it a perfect candidate for this purpose. 

\smparagraph{The algorithm.} Our main algorithm, which is formalized as Algorithm~\ref{alg:nearmaximal}, uses three randomization steps, all of which are necessary for the analysis: 
\begin{itemize}[topsep=5pt,itemsep=0ex,partopsep=0ex,parsep=1ex]
	\item An ordering $\pi$ over the edges is chosen uniformly at random.
	\item Each edge of the graph is sampled independently with some probability $p$.
	\item The vertex set $V$ is randomly partitioned into disjoint subsets $V_1, \ldots, V_k$.
\end{itemize}
After these steps, we put the edge-sampled induced subgraph of each $V_i$ into machine $i$ and compute a greedy maximal matching $M_i$ according to ordering $\pi$. The parameters $p,k$ are chosen to ensure that the induced subgraphs fit the memory of a machine. Observe that $M = \bigcup_{i\in[k]} M_i$ is a valid matching since the partitions are vertex-disjoint. 

\smparagraph{The analysis outline.} The key to our results, and the technical core of our paper, is to show that if we commit $M$ to the final maximal matching, then the degree of most vertices drops to $\Delta^{1-\Omega(1)}$ in the residual graph. 

For a vertex $v$ and partition $i \in [k]$,  let $Z_{v, i}$ denote the number of neighbors of $v$ in partition $i$ that are not matched in greedy matching $M_i$. So either $v$ is matched or its remaining degree in the residual graph becomes $\sum_{i \in [k]} Z_{v, i}$. Our degree reduction guarantee boils down to showing a concentration bound on the random variable $Z_{v, i}$.

Let us first outline how a concentration bound on $Z_{v, i}$ can be helpful. Suppose that, wishfully thinking, we have $Z_{v, i} = (1 \pm o(1)) \E[Z_{v, i}]$ for all $i$. By symmetry of the partitions,  $\E[Z_{v, i}] = \E[Z_{v, 1}]$ for each $i$. This means that all random variables $Z_{v, 1}, \ldots, Z_{v, k}$ take on the same values ignoring the lower terms. If $\E[Z_{v, 1}]$ is small enough that $k \cdot \E[Z_{v, 1}] < \Delta^{1-\Omega(1)}$, this gives the desired bound on the residual degree of $v$. Otherwise, if $\E[Z_{v,1}]$ is large, then $v$ would have many edges available in its own partition; the greedy matching would then be likely to choose one of these edges and so $v$ would itself be matched.

Unfortunately, there are two severe roadblocks to showing such concentration bounds. First, $Z_{v, i}$ is a rather complicated function of the underlying random variables. In general, concentration bounds obtained by Azuma's or other ``dimension dependent'' inequalities are not suitable for our purposes, as these would give bounds on the order of $Z_{v, i} = \E[Z_{v, i}] \pm \widetilde{O}(\sqrt{n})$, which is useless when $\Delta$ is much smaller than $n$. Recall that Chernoff-Hoeffding bounds, which are dimension-independent, apply to sums of independent random variables.  

The second roadblock is that we cannot even bound the variance of $Z_{v,i}$ for an arbitrary vertex $v$. The bounds we will develop are based on analysis of a random query process, for which we can only bound the \emph{average} over all vertices $v$. 

To handle both of these issues, we focus on an easier goal: instead of showing that the \emph{maximum degree} of the residual graph is reduced, we only show that the \emph{average degree} is reduced. This allows us to average out both types of fluctuation (within a vertex and across vertices). For this, instead of an exponential concentration bound, it suffices to use a weaker bound in terms of the variance of $Z_{v, i}$. We use a method known as the Efron-Stein inequality (see Proposition~\ref{prop:efronstein}); this is not among the standard tools used in theoretical computer science, and we hope this example shows that it can be useful in the analysis of randomized algorithms.

\smparagraph{How greedy maximal matching helps.} Our proof relies on a number of unique properties of the random greedy maximal matching algorithm:
\begin{enumerate}[topsep=8pt,itemsep=1ex,partopsep=0ex,parsep=1ex]
\item If we run greedy maximal matching on an edge-sampled subgraph of a graph, the maximum degree in the residual graph drops significantly.
\item The set of matched vertices in the greedy maximal matching changes by a constant number of elements if a single vertex or edge is modified.
\item If an edge ordering is chosen randomly, then determining whether an edge $e$ belongs to the greedy maximal matching requires ``looking back'' at just $O(d)$ other edges on average, where $d$ is the average degree of the line graph.
\end{enumerate}
We formalize these properties in Section~\ref{sec:greedy}. Note that Property 1 was the only property used in the maximal matching algorithm of Lattanzi \etal~\cite{DBLP:conf/spaa/LattanziMSV11}. Property 3 was originally developed in the context of {\em sublinear time algorithms} for approximate maximum matching. To our knowledge, it was first formalized by Nguyen and Onak~\cite{DBLP:conf/focs/NguyenO08}, with a precise bound we use shown by Yoshida \etal{}~\cite{DBLP:conf/stoc/YoshidaYI09}. It is remarkable that this methodology can be applied to concentration bounds.

\section{Preliminaries}

\subsection{Notation} 
For any integer $k$, we let $[k]$ denote the set $\{1, \ldots, k\}$. For a graph $G = (V, E)$ and vertex set $V' \subseteq V$, we let $G[V']$ denote the induced subgraph on $V'$. For a vertex $v$, we define the \emph{neighborhood} $N(v)$ to be the set of vertices $u$ with $\{u,v\} \in E$.

An edge subset $M \subseteq E$ is a {\em matching} if no two edges in $M$ share an endpoint. A matching $M$ of a graph $G$ is a {\em maximal matching} if it is not possible to add any other edge of $G$ to $M$. When it is clear from the context, we abuse notation to use $M$ for the vertex set of matching $M$; in particular, we write $G[V \setminus M]$ for the graph obtained by removing every vertex of $M$ from $G$. 

For any vertex $v \in V$ and a matching $M$, we define the \emph{residual degree} $\res_M(v)$ to be zero if $v \in M$, and otherwise $\res_M(v) = \deg_{G[ V \setminus M ]}(v)$. Finally, we define the \emph{match-status} of vertex $v$ according to matching $M$ to be the indicator for the event that $v \in M$.

\subsection{The \MPC{} Model}\label{sec:mpc}
The Massively Parallel Computations (\MPC{}) model was first introduced by Karloff \etal{}~\cite{DBLP:conf/soda/KarloffSV10} and further refined by \cite{DBLP:conf/isaac/GoodrichSZ11,DBLP:journals/jacm/BeameKS17, DBLP:journals/jacm/BeameKS17, DBLP:conf/stoc/AndoniNOY14}. An input of size $N$ is initially distributed among $M$ machines, each with a local space of size $S$. Computation proceeds in synchronous rounds in which each machine can perform an arbitrary local computation on its data and can send messages to other machines. The messages are delivered at the start of the next round. Furthermore, the total messages sent or received by each machine in each round should not exceed its memory.

We desire algorithms that use a sublinear space per machine (i.e., $S = N^{1-\Omega(1)}$) and only enough total space to store the input (i.e., $S \cdot M = O(N)$). For graph problems, the edges of an input graph $G = (V, E)$ with $n := |V|$ and $m := |E|$ are initially distributed arbitrarily among the machines, meaning that $N = \Theta(m)$ words (or $\Theta(m \log n)$ bits). We mainly consider the regime of \MPC{} with space per machine of $S = \Theta(n)$ words.

\subsection{Concentration inequalities}

We will use two main concentration inequalities: the \emph{Efron-Stein} inequality and the \emph{bounded differences inequality}. These both concern functions with certain types of Lipschitz properties.

\begin{proposition}[Efron-Stein inequality \cite{steele1986efron}]\label{prop:efronstein}
Fix an arbitrary function $f: \{0, 1\}^n \to \mathbb{R}$ and let $X_1, \ldots, X_n$ and $X'_1, \ldots, X'_n$ be $2n$ i.i.d. Bernoulli random variables. For $\vec{X} := (X_1, \ldots, X_n)$ and $\vec{X}^{(i)} := (X_1, \ldots, X_{i-1}, X'_i, X_{i+1}, \ldots, X_n)$, we have
$$\var(f(\vec{X})) \leq \frac{1}{2} \cdot \E \Big[ \sum_{i=1}^n \big(f(\vec{X}) - f(\vec{X}^{(i)})\big)^2\Big].$$
\end{proposition}

We use the following form of the bounded differences inequality (which is a special case of McDiarmid's inequality):

\begin{proposition}[Bounded differences inequality]\label{prop:bd-dif}
Let $f : \mathcal{X}_1 \times \ldots \times \mathcal{X}_k \to \mathbb{R}$ be a function on $k$ variables, such that changing any coordinate $x_i$ changes the function value $f(x_1, \dots, x_k)$ by at most $\lambda$. For a vector $\vec X = (X_1, \dots, X_k)$ of $k$ independent (not necessarily identically distributed or binary) random variables and for any $t \geq 0$ we have
$$
\Pr \big(f(\vec{X}) \geq \E[f(\vec{X})]  + t \big) \leq \exp\Big(\frac{-2 t^2}{k \lambda^2}\Big).
$$

In particular, w.e.h.p. there holds $
f(\vec{X}) \leq \E[f(\vec{X})] + \lambda n^{0.01} \sqrt{k}.
$ 
\end{proposition}

\subsection{Sequential Greedy Maximal Matching}
\label{sec:greedy}
As described in Section~\ref{sec:highlevel}, a maximal matching can be found by a sequential greedy algorithm:
\begin{definition}[Greedy maximal matching]
	Given a graph $G = (V, E)$ and an ordering $\pi$ over the edges $E$, the greedy maximal matching algorithm processes the edges in the order of $\pi$ and adds each edge $e$ to the matching if none of its incident edges have joined the matching so far. We denote the resulting maximal matching by $\greedymatching{G, \pi}$.
\end{definition}

We say that $e$ has \emph{higher priority than $e'$} if $\pi(e) < \pi(e')$, that is, $e$ is processed before $e'$. It is often convenient in this context to generate the permutation $\pi$ by choosing a random function $\rho: E \rightarrow [0,1]$, and then sorting in order of $\rho$. We write $\greedymatching{G, \rho}$ in this case as shorthand for $\greedymatching{G, \pi}$ where $\pi$ is the permutation associated to $\rho$.

This greedy maximal matching has a number of nice properties that play a critical role in the analysis of our algorithm; proofs appear in Appendix~\ref{app:greedy}. 

The first property is that runing the greedy matching on an edge-sampled subgraph will significantly reduce the degree of the residual graph. 

\begin{lemma}\label{lem:edgesamplegreedy}
Fix a graph $G = (V, E)$ and a permutation $\pi$ over $E$. Suppose we form an edge-set $L \subseteq E$ by sampling each edge with some given probability $p \in (0,1]$. Let $G' = (V, L)$ be the resulting subgraph and let $M := \greedymatching{G', \pi}$. Then for any vertex $v$ and any parameter $\beta > 0$, we have
$
\Pr \bigl( \res_M(v) \geq \beta/p \bigr) \leq e^{-\beta}.
$
\end{lemma}

The second useful property is that modifying a single vertex or edge of $G$ does not change the set of matched vertices too much. Note that the set of \emph{edges} in the matching can change significantly.

\begin{lemma}\label{lem:lipschitzproperties}
	Fix a graph $G=(V,E)$ and a mapping  $\rho: E \rightarrow [0,1]$.
	\begin{enumerate}[itemsep=-0.5ex]
		\item If graph $G'$ is derived by removing a vertex of $G$, then there are at most two vertices whose match-status differs in $\greedymatching{G, \rho}$ and $\greedymatching{G', \rho}$.
		\item If graph $G'$ is derived by removing an edge of $G$, then there are at most two vertices whose match-status differs in $\greedymatching{G, \rho}$ and $\greedymatching{G', \rho}$.
		\item If $\rho'$ is derived by changing a single entry of $\rho$, then there are at most two vertices whose match-status differs in $\greedymatching{G, \rho}$ and $\greedymatching{G, \rho'}$.
	\end{enumerate}
\end{lemma}

The third property can be summarized as stating that the presence of any given edge $e$ appearing in $M$ can be determined from a relatively small number of other edges. More precisely, we consider the following query-based method which we refer to as the ``edge oracle'' $\edgeoracle_{\pi}(e)$ to determine whether $e$ appears in $\greedymatching{G, \pi}$:

\begin{titledtbox}{$\edgeoracle_{\pi}(e)$: A query-process to determine whether  $e \in \greedymatching{G, \pi}$.}
	\begin{algorithmic}
		\State Let $e_1, \ldots, e_d$ be the incident edges to $e$ in $G$ sorted such that $\pi(e_1) < \pi(e_2) < \dots < \pi(e_d)$.
		\For{$i = 1, \ldots, d$}
			\If{$\pi(e_i)  < \pi(e)$}
			\State \textbf{if} $\edgeoracle_{\pi}(e_j) = \yes$ \textbf{then return} \no{}
			\EndIf
		\EndFor
		\State \Return \yes{}
	\end{algorithmic}
\end{titledtbox}
	
It is clear that $e \in \greedymatching{G, \pi}$ if and only if $\edgeoracle_{\pi}(e) = \yes{}$.   Translating a result of  Yoshida \etal{}~\cite{DBLP:conf/stoc/YoshidaYI09} for maximal independent set into our context gives:
\begin{proposition}[\cite{DBLP:conf/stoc/YoshidaYI09}]
\label{prop:greedmatchquery}
Fix a graph $G = (V,E)$ with $m$ edges and $r$ pairs of intersecting edges. For each edge $e \in E$, let $A(e)$ be the number of (recursive) calls to $\edgeoracle_{\pi}$ generated by running $\edgeoracle_{\pi}(e)$, including the original call to $\edgeoracle_{\pi}(e)$ itself. If $\pi$ is drawn uniformly at random from permutations on $m$ elements, then $\E_{\pi} [ \sum_{e \in E} A(e) ] \leq m + r$.
\end{proposition}


\section{Roadmap}
\label{sec:roadmap}
As discussed in Section~\ref{sec:highlevel}, the key to proving Theorem~\ref{thm:main} and Theorem~\ref{thm:sublinearspace} is an algorithm to reduce the graph degree by a polynomial factor. The precise statement of this lemma is as follows:

\begin{lemma}[degree reduction]
\label{lem:degreereduction}
There is an $O(1)$ round \MPC{} algorithm to produce a matching $M$, with the following behavior w.e.h.p.: it uses $O(n/\Delta^{\Omega(1)})$ space per machine and $O(m + n)$ space in total, and the residual graph $G[V \setminus M]$ has maximum degree $\Delta^{1-\Omega(1)}$.
\end{lemma}

The degree reduction algorithm of Lemma~\ref{lem:degreereduction} can immediately be used to prove Theorem~\ref{thm:main}.

\begin{proof}[Proof of Theorem~\ref{thm:main}, assuming Lemma~\ref{lem:degreereduction}]
The algorithm consists of $r$ iterations that each commits some edges to the final maximal matching using Lemma~\ref{lem:degreereduction}. Let $\Delta_i$ denote the maximum degree after $i$ iterations. In each iteration $i$, we reduce $\Delta_i$ to $\Delta_{i+1} \leq \Delta_i^{1-\alpha}$ for some constant $\alpha$.  So $\Delta_r \leq \Delta^{(1-\alpha)^{r}}$.  In particular, for $r = O( \log \log \Delta)$, the residual graph at the end has degree $O(1)$; we can store it on a single machine with $O(n)$ space, and compute its maximal matching. Similarly, for $r = \Theta( \log (1/\delta) )$, the residual graph at the end has degree $n^{\delta}$, and we can compute its maximal matching on a single machine with $O(n^{1+\delta})$ space.
\end{proof}

By combining our algorithm with a known technique for simulating \local{} algorithms, we also immediately get Theorem~\ref{thm:sublinearspace} which reduces the space per machine to $n/2^{\Omega(\sqrt{\log n})}$. 
\begin{proof}[Proof of Theorem~\ref{thm:sublinearspace}, assuming Lemma~\ref{lem:degreereduction}]
Let $\gamma$ be a given arbitrary constant.  For $\Delta \geq 2^{\sqrt{\log n}}$, the degree reduction algorithm of Lemma~\ref{lem:degreereduction} uses a space per machine of $O(n/\Delta^{\Omega(1)}) \leq n/2^{\Omega(\sqrt{\log n})}$. If we apply it for $O(\log \log \Delta)$ iterations, we can reduce to a residual graph $G'$ with maximum degree $\Delta' \leq 2^{ \sqrt{\log n}}$ w.e.h.p. while using $n/2^{\Omega(\sqrt{\log n})}$ space per machine.  (If $\Delta \leq 2^{\sqrt{\log n}}$ originally, then simply set $G' = G$.)
	
	At this point, we switch to a different algorithm: we simulate the \local{} maximal matching algorithm  \cite{DBLP:conf/focs/BarenboimEPS12} which has round complexity $t = O(\log \Delta' + \polyloglog n)$ on $G'$ and has success probability $1 - 1/\poly(n)$. In the language of \cite{DBLP:journals/corr/abs-1807-06701}, this is a \emph{state-congested \local{} algorithm}; the {\em blind coordination lemma} of \cite{DBLP:journals/corr/abs-1807-06701} states that such an algorithm can be simulated in the MPC model in $O(\frac{t}{\log_{\Delta'} n} + \log t) = O(\log \log \Delta + \log \log \log n)$ rounds and $n^{1 - \Omega(1)}$ space per machine and $n^{1 + \gamma/2}$ total space,  exclusive of the space needed to store $G'$ itself.
	
	 To amplify the success probability to w.e.h.p., we run $n^{\gamma/2}$ separate independent executions in parallel. This brings the total space (aside from the storage of $G'$) up to $n^{1+\gamma}$.
\end{proof}

The core of our analysis in proving Lemma~\ref{lem:degreereduction} lies in showing that the following Algorithm~\ref{alg:nearmaximal} significantly reduces the number of edges of $G$.  Throughout, we define parameters:
$$
p := \Delta^{-0.77}, \qquad \qquad k: = \Delta^{0.12}
$$

\begin{tboxalg}{}\label{alg:nearmaximal}
	\textbf{Input:} A graph $G = (V, E)$ with maximum degree $\Delta$.
	
	\textbf{Output:} A matching $M$ in $G$.
	
	\begin{enumerate}[label={(\arabic*)},ref={\arabic*}, topsep=5pt,itemsep=-0.3ex,partopsep=0ex,parsep=1ex, leftmargin=*]
		\item \textbf{Permutation:} Choose a permutation $\pi$ uniformly at random over the edges in $E$.\label{line:permutation}
		\item \textbf{Edge-sampling:} Let $G^L = (V, L)$ be an edge-sampled subgraph of $G$ where each edge in $E$ is sampled independently with probability $p$. \label{line:edgesampling}
		\item \textbf{Vertex partitioning:} Choose a function $\chi: V \rightarrow [k]$ uniformly at random and form associated vertex partition $V_1, \dots, V_k$ where $V_i = \{ v: \chi(v) = i \}$.\label{line:vertexpartitioning}
		\item Each machine $i \in [k]$ receives the graph $G^L[V_i]$ and finds the greedy maximal matching $M_i := \greedymatching{G^L[V_i], \pi}$.
		\item Return matching $M := \bigcup_{i=1}^{k} M_i$.
	\end{enumerate}
\end{tboxalg}

For simplicity, throughout we write $G_i$ for the graph $G[V_i]$ and $G^L_i$ for $G^L[V_i]$. We let $L_i$ be the set of edges $\{u,v\} \in L$ with $u,v \in V_i$; that is, $L_i$ is the edge-set of $G_i^L$.  

The following result summarizes Algorithm~\ref{alg:nearmaximal}:

\begin{lemma}\label{lem:nearmaximal}
Algorithm~\ref{alg:nearmaximal} has the following desirable behavior:
\begin{enumerate}[itemsep=-0.5ex]
\item W.e.h.p.,  it uses $O(n/\Delta^{\Omega(1)})$ space per machine.
  \item W.e.h.p., it uses $O(n + m/\Delta^{\Omega(1)})$ space in total (aside from storing the original input graph.)
  \item The expected number of edges in the residual graph is at most $O(n \Delta^{0.89})$.
\end{enumerate}
\end{lemma}

The first two parts of Lemma~\ref{lem:nearmaximal} are straightforward consequences of the randomization.
 \begin{proof} [Proof of Lemma~\ref{lem:nearmaximal} part 1 and 2]
 First, by a straightforward Chernoff bound, we have $|V_i| = \Theta(n/k)$ w.e.h.p. for all $i \in [k]$; here note that $\E[V_i] = n/k = n/\Delta^{0.12} \geq \poly(n)$. 
 
 We next claim that $|L_i| \leq O(n \Delta p / k^2)$ for all $i$. For, consider random variable $Y_i = |L_i|$. For each vertex $v$, an incident edge $e = \{u, v\}$ will belong to $L_i$ if $e$ is sampled in $L$ and  vertex $u$ also belongs to $V_i$. So the expected number of neighbors of $v$ in $G^L_i$ is at most $\Delta p / k = \Delta^{0.11}$.  If $\Delta \geq n^{0.1}$, then a simple Chernoff bound shows that the number of neighbors of each vertex $v$ in each $G_i^L$ is concentrated around $O(\Delta p/k)$ w.e.h.p. Combined with the bound on $|V_i|$, this implies that $Y_i \leq O(n/k) \cdot O(\Delta p/k) = O(n \Delta p/k^2)$ w.e.h.p.

Otherwise, if $\Delta < n^{0.1}$, then observe that $Y_i$ can be regarded as a function of the vertex partition $\chi$ and the edge set $L$. There are $O(n \Delta)$ such random variables, and each of these changes $Y_i$ by at most $\Delta$. Also $\E[Y_i] \leq n \Delta p /k^2 = n \Delta^{-0.01}$.  Therefore, by Proposition~\ref{prop:bd-dif}, w.e.h.p., we have $Y_i \leq \E[Y_i] + \Delta \cdot n^{0.01} \cdot \sqrt{O(n \Delta)}$;  as $\Delta \leq n^{0.1}$ this implies that $Y_i \leq O( n \Delta p / k^2) $ w.e.h.p.

So each machine $i$ requires space of $|V_i| = O(n/k) = O(n \Delta^{-0.12})$ for its vertices and $|L_i| = O(n \Delta p/k^2) = O(n \Delta^{-0.01})$ for its edges. To show the bound on total space usage,  note that the total edge count of all the graphs $G^L_i$ is at most $|L|$, since each edge lives on at most one machine. We have $\E[ |L| ] = m p$ so  a straightforward  Chernoff bound gives $|L| \leq  O(m p + n) = O(m \Delta^{-0.77} + n)$ w.e.h.p. Furthermore, storing the vertex partition $\chi$ requires only $O(n)$ total space. 
\end{proof}

The third part of Lemma~\ref{lem:nearmaximal} is the hard part; we show it next in Section~\ref{sec:nearmaximal}. In  Section~\ref{sec:leftoververtices}, we add a few post-processing steps after Algorithm~\ref{alg:nearmaximal} which give Lemma~\ref{lem:degreereduction}. 

\section{Proof of Lemma~\ref{lem:nearmaximal} Part 3}\label{sec:nearmaximal}
For any vertex $v \in V$ and any $i \in [k]$, we consider the random variable
$$
Z_{v, i} := \bigl| V_i \cap N_{G[V \setminus M]}(v) \bigr|,
$$ that is, the number of neighbors of $v$ which are unmatched in $G_i$. Here $v$ does not necessarily belong to $V_i$. In particular, if $v$ is not matched in $M$, we have $\res_M(v) = Z_{v, 1} + \dots + Z_{v, k}$. We further define the related random variable $Z'_v$ as:
$$
Z'_v :=  \begin{cases}
Z_{v, \chi(v)} & \text{if $v \notin M$} \\
0 & \text{if $v \in M$,}
\end{cases}
$$
which is equivalent to the residual degree of $v$ in its own partition. 

The key to analyzing Algorithm~\ref{alg:nearmaximal}, as sketched in Section~\ref{sec:highlevel}, is to show that for most vertices $v$, the values $Z_{v,i}$ are nearly equal across all indices $i$.  The following claims make this precise.

\begin{claim}\label{cl:intrapartitiondegree}
For any vertex $v$ and any parameter $\beta \geq 0$, we have  $\Pr( Z'_v \geq \beta/p ) \leq e^{-\beta}$.
\end{claim}
\begin{proof}
We claim this bound holds, even after conditioning on random variables $\chi$ and $\pi$. For, suppose that $\chi(v) = i$. Here $M_i$ is formed by performing independent edge sampling on $G[V_i]$ and then taking the greedy maximal matching. Thus by Lemma~\ref{lem:edgesamplegreedy}, the probability that $v \notin M_i$ and $v$ has more than $\beta/p$ unmatched neighbors in $G_i$ is at most $e^{-\beta}$. 
\end{proof}

\begin{claim}
\label{var-bnd-lemma}
For any vertex $v$, suppose we condition on the random variables $\pi, L$, and we define the related random variable
$$
\sigma_{v \mid \pi, L}  = \sqrt{ \var(Z_{v,1} \mid \pi, L) }
$$

Then for any parameter $\alpha > 0$, we have 
$$
\Pr( \res_M(v) > k Z'_v + \alpha k) \leq O \Bigl( \frac{k \sigma_{v \mid \pi, L}^2}{ \alpha^2} \Bigr),
$$
 where all probabilities are taken with respect to the remaining random variable $\chi$.
\end{claim}
\begin{proof}
Write  $\sigma = \sigma_{v \mid \pi, L}$ and $\mu = \E[ Z_{v,1}]$ for brevity; by symmetry of the partitions, we have $\var(Z_{v,i}) = \sigma^2$ and $\E[Z_{v,i}] = \mu$ for any index $i$. Chebyshev's inequality immediately gives 
$$
\Pr(  | Z_{v,i} - \mu | > \alpha/2 ) \leq O ( \sigma^2 / \alpha^2 ).
$$
By a union bound over all indices $i \in [k]$, there is a probability of at least $1 - O ( k \sigma^2 / \alpha^2 )$ that $| Z_{v,i} - \mu | \leq \alpha/2$ for all $i$.   Now suppose this event has occurred. If $v$ is matched, then $\res_M(v) = Z'_v = 0$ and clearly $\res_M(v) \leq k Z'_v + \alpha k$. Otherwise,  $\res_M(v) = Z_{v,1} + \dots + Z_{v,k} \leq k \mu + k \alpha/2$ and $Z'_v = Z_{v, \chi(v)} \geq \mu - \alpha/2$, and so $\res_M(v) - k Z'_v \leq  (k \mu + k \alpha/2) - k (\mu - \alpha/2) = k \alpha$ as desired.
\end{proof}

We can combine these two estimates in the following elegant result:
\begin{lemma}
\label{var-bnd-cor}
For a vertex $v$, there holds
$$
\E[ \res_M(v) ] \leq O \bigl( k/p + k^{3/2}  \E[ \sigma_{v \mid \pi, L} ]   \bigr)
$$
where expectations are taken over all random variables $\pi, L, \chi$.
\end{lemma}
\begin{proof}
Write $\sigma = \sigma_{v \mid \pi, L}$ for brevity. If we condition on a fixed $\pi, L$, then using integration by parts, we get:
$$
\E[ \res_M(v) \mid \pi, L ] \leq k \E[  Z'_v \mid \pi, L ] + k^{3/2} \sigma +  \int_{x = k^{3/2} \sigma}^{\infty} \Pr( \res_M(v) \geq k Z'_v + x  \mid \pi, L) \ dx.
$$

We can apply Claim~\ref{var-bnd-lemma} with respect to $\alpha = x/k$, getting
\begin{align*}
& \int_{x = k^{3/2} \sigma}^{\infty} \Pr( \res_M(v) \geq k Z'_v + x  \mid \pi, L) \ dx \leq \int_{x = k^{3/2} \sigma}^{\infty} O( k^3 x^{-2} \sigma^2  ) \ dx  = O( k^{3/2} \sigma ) .
\end{align*}

Integrating over the random variables $\pi, L$ then gives:
$$
\E[ \res_M(v)  ]  \leq k \E[Z'_v] + O( k^{3/2} \E_{\pi, L} [ \sigma ]   ).
$$

To compute $\E[Z'_v]$, we again use integration by parts and apply Claim~\ref{var-bnd-lemma} to get
\[
\E[Z'_v] = \int_{x = 0}^{\infty} \Pr(Z'_v \geq x) \ dx \leq \int_{x = 0}^{\infty} e^{-x p} \ dx = 1/p. \qedhere
\]
\end{proof}

\subsection{Analysis of the Variance}
In light of Lemma~\ref{var-bnd-cor}, we need to bound the variance of $Z_{v,1}$ as a function of the vertex partition $\chi$. So let us  assume that $L$ and $\pi$ have been fixed, and $Z_{v, 1}$ is a function of $\chi$ alone, or more precisely, a function of the set of vertices in partition $V_1$. Define a vector $\vec x$ by setting $x_v = 1$ if $\chi(v) = 1$, and $x_v = 0$ otherwise; we may write $Z_{v, 1}(\vec{x})$ to emphasize that $Z_{v, 1}$ is merely a function of $\vec{x}$. Observe that $\vec{x}$ is a vector of $n$ i.i.d. Bernoulli-$1/k$ random variables. To use the Efron-Stein inequality for bounding the variance, we have to upper bound the right-hand-side of inequality
\begin{equation}\label{eq:efronvertex}
	\sigma_{v \mid \pi, L}^2 \leq \frac{1}{2}\E_{\vec{x}}\Big[\sum_{w \in V} \big(Z_{v, 1}(\vec{x}) - Z_{v, 1}(\vec{x}^{(w)})\big)^2\Big],
\end{equation}
where $\vec{x}^{(w)}$ is obtained by replacing the value of $x_w$ in $\vec{x}$ with $x'_w$ which is drawn independently from the same distribution. In other words, the $w$ summand of Eq.~(\ref{eq:efronvertex}) corresponds to the effect of repartitioning vertex $w$ on the value of $Z_{v, 1}$. 

As a starting point, we note a Lipschitz property coming from the greedy maximal matching.
\begin{claim}[Lipschitz property]\label{cl:lipschitz} There holds $(Z_{v, 1}(\vec{x}) - Z_{v, 1}(\vec{x}^{(w)}))^2 \leq 4$.
\end{claim}
\begin{proof}
 Let $V_1$ and $V'_1$ denote the vertex partitions due to $\vec{x}$ and $\vec{x}^{(w)}$ respectively, i.e., $V_1 = \{ u \mid x_u = 1 \}$ and $V'_1 = \{ u \mid x^{(w)}_u = 1 \}$. These partitions differ in at most one vertex, namely, $w$. Correspondingly, define $M_1 := \greedymatching{G[V_1], \pi}$ and $M'_1 := \greedymatching{G[V'_1], \pi}$. By Lemma~\ref{lem:lipschitzproperties} part 1, there are at most two vertices in $V$ whose match-status differs between $M_1$ and $M'_1$.
\end{proof}

Note that $\Pr( \vec x \neq {\vec x}^{(w)}) \leq 2/k$. Thus, Claim~\ref{cl:lipschitz} gives $\E[ (Z_{v,1} - Z_{v,1}(\vec{x}^{(w)}) )^2 ] \leq 8/k$, and Eq.~(\ref{eq:efronvertex}) in turn implies $\E[\var(Z_{v,1})] \leq 4n /k$. To get a tighter bound independent of $n$, we analyze a type of ``query process'' to determine $Z_{v,1}$ while only examining a subset of the entries of $\vec x$.

\smparagraph{The query process.} We can use a recursive query process, which we denote by  $\edgeoracle_{\pi}(e, \vec{x})$, to determine whether a given edge $e$ belongs to the matching $M_1(\vec{x})$ --- here we emphasize that $\pi, L$ should be regarded as fixed. This is very similar to the edge oracle for greedy matching  discussed in Section~\ref{sec:greedy}, except that instead of querying edges, it queries the entries of the vector $\vec{x}$.

\begin{titledtbox}{$\edgeoracle_{\pi}(e, \vec{x})$: A query-process to determine whether  $e \in M_1(\vec{x})$.}
	\begin{algorithmic}
		\State Let $e = \{u, v\}$. Query $x_u$ and $x_v$; \textbf{if} $x_u = 0$ or $x_v = 0$, \textbf{then} \Return \no.
		\State Let $e_1, \ldots, e_d$ be the incident edges to $e$ in $G^L$ sorted as $\pi(e_1) < \pi(e_2) < \dots < \pi(e_d)$.
		\For{$i = 1, \ldots, d$}
			\If{$\pi(e_i)  < \pi(e)$}
				\State	\textbf{if} {$\edgeoracle_{\pi}(e_j, \vec{x}) = \yes$} \textbf{then return} \no{}
			\EndIf
		\EndFor
		\State \Return \yes{}
	\end{algorithmic}
\end{titledtbox}

We also define a {\em degree oracle} $\degreeoracle_{\pi}(v, \vec{x})$ to determine the value of $Z_{v, 1}(\vec{x})$. This checks whether each neighbor $u$ of $v$ appears in $V_1$ and is matched, which in turn requires checking whether every edge incident to $u$ appears in matching of $G^L[V_1]$:

\begin{titledtbox}{$\degreeoracle_{\pi}(v, \vec{x})$: A query process to determine the value of $Z_{v, 1}(\vec{x})$.}
	\begin{algorithmic}
		\State $c \gets 0$
		\For{all vertices $u \in N_G(v)$}
			\State Query $x_u$.
			\If{$x_u = 1$}
				\State Execute $\edgeoracle_{\pi}( \{ u, w \}, \vec{x})$ for all vertices $w \in N_{G^L}(u)$.
	\State					\textbf{if}  $\edgeoracle_{\pi}( \{ u, w \}, \vec{x}) = \no{}$ for all such vertices $w$ \textbf{then} $c \gets c+1$ \Comment{$u$ is unmatched in $M_1$}
			\EndIf
		\EndFor
		\State \Return $c$
	\end{algorithmic}
\end{titledtbox}

We now analyze the {\em query complexity} $B(v)$ of the oracle $\degreeoracle_{\pi}$, i.e., the number of indices in $\vec{x}$ that are queried when running $\degreeoracle_{\pi}(v)$. 
\begin{claim}\label{cl:varboundedbyB}
For fixed $\pi, L$ we have $\sigma_{v \mid \pi, L}^2 \leq 2 B(v)$.
	\end{claim}
	\begin{proof}
	
By definition, the value of $Z_{v, 1}(\vec{x})$ can be uniquely determined by revealing the indices of $\vec{x}$  quered by $\degreeoracle_\pi(v, \vec{x})$. So changing other indices $w$ of $\vec{x}$ cannot affect $Z_{v, 1}$ and hence $Z_{v, 1}(\vec{x}) = Z_{v, 1}(\vec{x}^{(w)})$ for such $w$. There are $B(v)$ indices queried by $v$; for each such index $w$, Claim~\ref{cl:lipschitz} gives $(Z_{v, 1}(\vec{x}) - Z_{v, 1}(\vec{x}^{(w)}))^2 \leq 4$. Overall, we get 
	$$
	\sum_{w \in V} (Z_{v, 1}(\vec{x}) - Z_{v, 1}(\vec{x}^{(w)}))^2 \leq 4 B(v).
	$$
	
By  the Efron-Stein inequality (\ref{eq:efronvertex}) this immediately implies $\sigma_{v \mid \pi, L}^2 \leq 2 B(v).$
\end{proof}

To bound $B(v)$, let us first define $A_1(e)$ for an edge $e \in L_1$ to be the number of edges in $L_1$ on which the edge oracle is called (recursively) in the course of running $\edgeoracle_\pi(e, \vec{x})$; this includes edge $e$ itself. We similarly define $A_1(v)$ to be the number of edges in $L_1$ that are queried in the course of running $\degreeoracle_{\pi}(v)$.  Note that when running $\edgeoracle_\pi(e, \vec{x})$ or $\degreeoracle_{\pi}(v)$, only edges  in $L_1$ can generate new recursive calls; other edges are checked, but immediately discarded.

\begin{claim}
\label{cl:l1tol}
For any vertex $v$, we have $\E[B(v)] \leq O(\Delta^2 p + \E[ A_1(v) ] \Delta p)$,  where both expectations are taken over $\chi$, $L$, and $\pi$.
\end{claim}
\begin{proof}
First,  $\degreeoracle_{\pi}(v, \vec x)$ will query each vertex $u \in N_G(v)$, and if $u \in G_1$ it will then query all vertices $w \in N_{G^L}(u)$. Thus, there are at most $\sum_{u \in N_G(v)} (1 + \deg_{G^L}(u))$ queried vertices directly produced in $\degreeoracle_{\pi}(v, \vec x)$.  This has expectation at most $\Delta (1 + \Delta p)$, which is $O(\Delta^2 p )$ since $\Delta p \geq 1$.

Next, let us count the queries produced recursively through calls to $\edgeoracle_{\pi}(e, \vec{x})$.  Let $J$ denote the set of edges in $L_1$ queried during execution of $\degreeoracle_{\pi}(v)$, where $|J| = A_1(v)$ by definition.  Suppose we condition on the random variables $\chi, L_1$ and $\pi$; this determines the set $J$.  The only randomness remaining is to determine whether each edge $e \in G_2 \cup \dots \cup G_k$ goes into $L$. 

For each edge $e =\{u,w \} \in J$, the execution of $\edgeoracle_{\pi}(e, \vec{x})$ calls $\edgeoracle_{\pi}(f, \vec{x})$ for edges $f \in L \setminus L_1$ which touch $e$; each of these will query two vertices, but the query process will not proceed further when they are discovered to lie outside $L_1$. So $e$ incurs $2( \deg_{L \setminus L_1}(u) + \deg_{L \setminus L_1}(w))$ additional queries, which has expectation at most $4 \Delta p$.  

Thus for fixed $\chi, L_1, \pi$, the expected number of resulting queries is $ |J| (1 + 4\Delta p)$.  Integrating over $\chi, L_1, \pi$, and noting that $\Delta p \geq 1$, gives the claimed bound.
\end{proof}

\begin{claim}
\label{claim:aebound}
We have $\E [ \sum_{e \in L_1} A_1(e) ] \leq O( n \Delta p / k^2 + n \Delta^2 p^2 / k^3 )$ where the expectation is taken over $\chi$, $L$, and $\pi$.
\end{claim}
\begin{proof}
Let us first suppose that $L$ and $\chi$ are fixed, and so $G_i^L$ is fixed as well. The only randomness remaining is the permutation $\pi$.  We are only interested in edges of $L_1$, so the edges outside $L_1$ have no effect on the behavior of $\edgeoracle_{\pi}$. Thus, $A_1(e)$ is precisely the query complexity of $e$ for a greedy matching of $G_1^L$ under a random permutation $\pi$. By Proposition~\ref{prop:greedmatchquery}, we have:
$$
\E_{\pi} \Big[ \sum_{e \in L_1} A_1(e) \mid L, \chi \Big] \leq |L_1| + |R_1|,
$$
where $R_1$ is the set of intersecting edge pairs in $G_1^L$. Integrating over random variables $L$ and $\chi$ gives:
$$
\E \Big[ \sum_{e \in L_1} A_1(e) \Big] \leq \E[ |L_1| + |R_1| ].
$$

 Each edge $e \in E$ goes into $L_1$ with probability $p/k^2 $,  so $\E[ |L_1| ] = m p/k^2$. Likewise, $G$ has at most $m \Delta / 2$ pairs of intersecting edges and each of these survives to $R_1$ with probability $p^2/k^3$. So, $\E[ |R_1| ] \leq m \Delta p^2/k^3$. Finally, we observe that $m \leq n \Delta$.
\end{proof}

We now can bound the average value of $B(v)$.
\begin{claim}\label{lem:boundonA}
We have $\tfrac{1}{n} \E[ \sum_{v \in V} B(v) ]  \leq O (\Delta^4 p^3 /k^3 + \Delta^3 p^2 / k^2+ \Delta^2 p)$ where the expectation is taken over $\chi, L, \pi$.
\end{claim}
\begin{proof}
For any vertex $v$, observe that $$
A_1(v) =  \sum_{\substack{u \in N(v), w \in N(u) \\ \{u, w \} \in L_1}}  A_1( \{u, w \}).
$$
 Summing over $v \in V$, we get:
\begin{align*}
\sum_{v \in V} A_1(v) &= \sum_{ \{u,w \} \in L_1} A_1( \{u,w \} ) \Bigl( \sum_{v \in N(u)} 1 + \sum_{v \in N(w)} 1 \Bigr) \leq 2 \Delta \sum_{e \in L_1} A_1(e).
\end{align*}
Taking expectations and applying Claim~\ref{claim:aebound}, we therefore have
$$
\E \big[ \sum_{v \in V} A_1(v) \big] \leq 2 \Delta \E \big[ \sum_{e \in L_1} A_1(e) \big] \leq O( n \Delta^2 p / k^2 + n \Delta^3 p^2 / k^3  ).
$$

Next applying Claim~\ref{cl:l1tol} gives
\[
\E \big[ \sum_{v \in V} B(v) \big] \leq O\big( \Delta p \E \big[ \sum_{v \in V} A_1(v) \big] + n \Delta^2 p \big) \leq O\big(n \Delta^3  p^2 / k^2+ n \Delta^4 p^3 /k^3 + n  \Delta^2 p \big).  \qedhere
\]
\end{proof}

Finally, we combine everything to get our bound on the average expected residual degree
\begin{claim}\label{lem:boundonA2} We have $
\tfrac{1}{n} \sum_{v \in V} \E[ \res_M(v) ] \leq O( k / p + \Delta^2 p^{3/2} + \Delta^{3/2} k^{1/2} p + \Delta k^{3/2} p^{1/2}  ) 
$
where expectations are taken over $\chi, L, \pi$.
\end{claim}
\begin{proof}
We can sum over vertices $v \in V$ and apply Lemma~\ref{var-bnd-cor} to get
$$
\sum_{v \in V} \E[ \res_M(v) ] \leq \sum_{v \in V} O \big( k/p + k^{3/2} \E_{\pi, L}[ \sigma_{v \mid \pi, L} ]\big)
$$

By Claim~\ref{cl:varboundedbyB}, we have $\sigma_{v \mid \pi, L}^2 \leq 2 B(v)$ for fixed $\pi, L$.  Taking expectations over $\pi, L$ gives:
$$
\sum_{v \in V} \E[ \res_M(v) ] \leq O \Bigl( n k / p + k^{3/2} \sum_{v \in V} \E[ \sqrt{B(v)} ] \Bigr).
$$

By Jensen's inequality, we have $\E[ \sqrt{B(v)} ] \leq \sqrt{\E[B(v)]}$ for any vertex $v$. Again by Jensen's inequality, we have 
$$
\tfrac{1}{n} \sum_{v \in V}  \sqrt{\E[B(v)]} \leq \sqrt{  \tfrac{1}{n} \sum_{v \in V} \E[B(v)] }
$$

 Claim~\ref{lem:boundonA} gives an upper bound on the sum $\tfrac{1}{n} \sum_v \E[B(v)]$; after collecting terms, this gives the claimed result.
\end{proof}

Plugging in the values $p,k$ gives $\sum_{v \in V} \E[ \res_M(v) ] \leq O(n \Delta^{0.89})$. This concludes the proof of   Lemma~\ref{lem:nearmaximal} part (3).

\section{Putting Everything Together}\label{sec:leftoververtices}
To finish Lemma~\ref{lem:degreereduction}, we need to remove all the remaining high-degree vertices, not just a (large) fraction of them.  Algorithm~\ref{alg:leftoververtices} handles these clean-up steps, with some additional implementation details discussed below.

\begin{tboxalg}{}\label{alg:leftoververtices}
	\begin{enumerate}[label={(\arabic*)},ref={\arabic*-}, topsep=0pt,itemsep=0ex,partopsep=0ex,parsep=1ex, leftmargin=*]
	\item Generate a matching $M$ such that the residual graph has at most $n \Delta^{0.9}$ edges  (see below)
		\item Sample each edge of $E$ with probability $q = \Delta^{-0.91}$ and let $L$ be the set of sampled edges.
		\item Put $G' = (V \setminus M, L)$ in machine 1, choose an arbitrary permutation $\pi$ over its edges and return matching $M' := \greedymatching{G', \pi}$.
		\item Let $U$ be the set of vertices $u$ with $\res_{M \cup M'}(u) \geq \Delta^{0.92}$,  and let $F$ be the set of edges with at least one endpoint in $U$.
\item Put $G'' = (V \setminus (M \cup M'), F)$ in machine 1, and choose an arbitrary maximal matching $M''$ of $G''$.
\item Return matching $M \cup M' \cup M''$.
	\end{enumerate}
\end{tboxalg}

\begin{claim}
\label{cl:alg1whp}
Step (1) of Algorithm~\ref{alg:leftoververtices} can be implemented to succeed w.e.h.p. using $O(n/\Delta^{\Omega(1)})$ space per machine and $O(m + n)$ total space.
\end{claim}
\begin{proof}
  We assume the original graph has $m \geq n \Delta^{0.9}$ edges as otherwise there is nothing to do.  We also assume that $n, \Delta$ are larger than any needed constants; if not, the entire graph can be put on a single machine in $O(n)$ space and maximal matching (or any other problem) can be solved trivially in a single round. 
  
  Now consider running Algorithm~\ref{alg:nearmaximal} to generate a matching $M$, and let $X$ be the number of edges in the residual graph. Lemma~\ref{alg:nearmaximal} has shown that $\E[X] \leq O(n \Delta^{0.89}) \leq n \Delta^{0.9} / 2$ for large enough $\Delta$, and so we need to show concentration for $X$.  There are two cases depending on the size of $\Delta$.

\smparagraph{Case 1: $\pmb{\Delta > n^{0.1}}$.} Markov's inequality applied to $X$ gives $\Pr[X > n \Delta^{0.9}] \leq 1/2$. We can run $t = n^a$ parallel iterations of Algorithm~\ref{alg:nearmaximal} for some constant $a > 0$, generating matchings $M_1, \dots, M_t$. Since they are independent, there is a probability of at least $1 - 2^{-t} = 1 - e^{-\poly(n)}$ that at least one matching $M_i$ has $X < n \Delta^{0.9}$ as desired. Each application of Algorithm~\ref{alg:nearmaximal} separately uses $O(n + m/\Delta^{\Omega(1)})$ space. Therefore, the $t$ iterations in total use $O(n^{1+a} +  n^a m/\Delta^{\Omega(1)})$ space. Since $\Delta > n^{0.1}$ and $m \geq  n \Delta^{0.9}$, this is $O(m)$ for sufficiently small constant $a$.

\smparagraph{Case 2: $\pmb{\Delta < n^{0.1}}$.} We can regard $X$ as being determined by $O(n \Delta)$ random variables, namely, $\rho, \chi, L$.  By Lemma~\ref{lem:lipschitzproperties}, modifying each entry of $\rho, \chi$, or $L$ only changes the match-status of $O(1)$ vertices. Each such vertex,  in turn, has only $\Delta$ neighbors, which are the only vertices whose degree in $G[ V \setminus M]$ is changed. Thus, changing each of the underlying random variables can only change $X$ by $O( \Delta^2 )$. By Proposition~\ref{prop:bd-dif}, therefore, w.e.h.p. we have
$$
X \leq \E[X] +  O(\Delta^2) n^{0.01} \sqrt{ n \Delta }  \leq n \Delta^{0.9}/2 + O(n^{0.51} \Delta^{2.5}).
$$
As $\Delta \leq n^{0.1}$, this is at most $n \Delta^{0.9}$ for large enough $n$.
\end{proof}

\begin{claim}
\label{cl:left1}
The vertex set $U$ at step (4) satisfies $|U| \leq O(n/\Delta^{1.01})$ w.e.h.p.
\end{claim}
\begin{proof}
By Lemma~\ref{lem:edgesamplegreedy} with $\beta = \Delta^{0.01}$, any vertex $v$ has $$
\Pr( \res_{M \cup M'}(v) \geq \Delta^{0.92} ) \leq e^{-\Delta^{0.01}}.
$$
 So, letting $Y = |U|$, we have $\E[ Y ] \leq n e^{-\Delta^{0.01}}.$ If $\Delta > n^{0.1}$, this already implies by Markov's inequality that $Y < 1$ w.e.h.p.  Otherwise, if $\Delta < n^{0.1}$, then we use the bounded differences inequality. Here, $Y$ can be regarded as a function of $n \Delta$ random variables, namely, the membership of each edge in $L$. By Lemma~\ref{lem:lipschitzproperties}, each edge affects the match-status of $O(1)$ vertices, and hence can change $Y$ by at most $O(\Delta)$. By Proposition~\ref{prop:bd-dif}, we therefore have w.e.h.p.
$$
Y \leq \E[Y] + O( \Delta \cdot \sqrt{n \Delta} \cdot n^{0.01}) \leq n e^{-\Delta^{0.01}} + O( \Delta^{1.5} n^{0.51}).
$$
By our assumption that $\Delta \leq n^{0.1}$, this is $O(n/\Delta^{1.01})$.
\end{proof}

\begin{claim}
Algorithm~\ref{alg:leftoververtices} uses $O(n/\Delta^{\Omega(1)})$ space per machine and total space $O(m + n)$ w.e.h.p. At the end of the process, the maximum degree of $G[V \setminus (M \cup M' \cup M'')]$ is at most $\Delta^{0.92}$.
\end{claim}
\begin{proof}
Claim~\ref{cl:alg1whp} shows these bounds hold for step (1). For step (3), the graph $G[V \setminus M]$  has at most $n \Delta^{0.9}$ edges, so $\E[ |L| ] \leq n \Delta^{0.9} q  \leq O(n/\Delta^{0.01})$; then a simple Chernoff bound shows $|L| \leq O(n/\Delta^{0.01})$ w.e.h.p,  so $G'$ can be stored on a single machine. For step (5), observe that $|F| \leq |U| \cdot \Delta \leq O(n/\Delta^{0.01})$ by Claim~\ref{cl:left1}; thus, again $G''$ can be stored onto a single machine.

Since $M''$ is a maximal matching of $G''$, all remaining vertices of $G[V \setminus (M \cup M' \cup M'')]$ must have degree at most $\Delta^{0.92}$.
\end{proof}

In particular, Algorithm~\ref{alg:leftoververtices} satisfies the claim of Lemma~\ref{lem:degreereduction}.

\appendix

\section{Useful Properties of Sequential Greedy Maximal Matching}
\label{app:greedy}
We show here the properties of the sequential greedy maximal matching used in the paper.

\begin{proof}[Proof of Lemma~\ref{lem:edgesamplegreedy}]
Consider the following equivalent method of generating $M$.  We iterate over the edges in $E$ in the order of $\pi$. Upon visiting an edge $e$, if one of its incident edges belongs to $M$, we call it {\em irrelevant} and discard it. Otherwise, we draw a Bernoulli-$p$ random variable $X_e$; if $X_e = 1$, we call $e$ {\em lucky} and add it to $M$ otherwise we call $e$ \emph{unlucky}.

If $v$ is matched in $M$, then $\res_M(v) = 0$. Otherwise, all of its remaining edges in $G[V \setminus M]$ should have been unlucky. That is, every time we encounter an edge $e$ in this process, it must have been irrelevant or we must have chosen $X_e = 0$. Furthermore, in order to get $\res_M(v) > \beta/p$,  there must remain at least $\beta/p$ edges which are not irrelevant. During this process, the probability that all such edges are marked unlucky is at most $(1-p)^{\beta/p} \leq e^{-\beta}$.
\end{proof}

\begin{proof}[Proof of Lemma~\ref{lem:lipschitzproperties}]
We start with the proof of the first part. Suppose that $G'$ is obtained by removing some vertex $v$ from $G$. Let $M := \greedymatching{G, \rho}$ and $M' := \greedymatching{G', \rho}$ and let $D := M \oplus M'$ denote the symmetric difference of $M$ and $M'$, i.e. $D = (M \setminus M') \cup (M' \setminus M)$. Note that the match-status of a vertex differs in $M$ and $M'$ if and only if its degree in $D$ is one. Therefore, it suffices to show that there are at most two such vertices in $D$.
		
We claim that $D$ has at most one connected component (apart from isolated vertices). For sake of contradiction, suppose $D$ has some non-trivial component $C$ which does not include vertex $v$. Let $e$ be the edge in $C$ with the highest priority, and suppose that $e \in M, e \notin M'$ (the case where $e \in M'$ is similar). Here $e \in G'$ since $v$ is not an endpoint of edge $e$. So by definition of the greedy matching, there must be some edge $f \in M'$ connected to $e$ with higher priority. However, since $e \in M$, it must be that $f \notin M$ and hence $f \in M \oplus M'$. So $f \in C$, contradicting our choice of $e$. 

Now $D$ is composed of the edges of two matchings, so its unique component is either a path or a cycle. The latter has no vertex of degree one and the former has two; proving part 1 of Lemma~\ref{lem:lipschitzproperties}.

The other two parts of Lemma~\ref{lem:lipschitzproperties} follows from a similar argument. If an edge $e$ is removed from $G$ or its entry in $\rho$ is changed, then again the symmetric difference $M \oplus M'$ of the matchings would contain only one connected component which has to contain $e$. Since this component is a cycle or a path, there are at most two vertices whose match-status differs in the two matchings.
\end{proof}

\begin{proof}[Proof of Proposition~\ref{prop:greedmatchquery}]
Let $H$ be the line graph of $G$, so $H$ has $m$ vertices and $r$ edges. Also, $\E_{\pi}[A(e)]$ is the expected query complexity of the greedy maximal independent set of $H$ under a random permutation. The result~\cite[Theorem 2.1]{DBLP:conf/stoc/YoshidaYI09} bounds the average value of $A(e)$ in terms of vertex and edge counts of $H$ as $\frac{1}{m} \E_{\pi} \big[ \sum_{e \in E} A(e) \big] \leq 1 + \frac{r}{m}$. We obtain the stated result by multiplying through by $m$.
\end{proof}

\bibliographystyle{plainnat}
\bibliography{referencesMPC}
	
\end{document}